\documentclass[draft]{amsart}
\usepackage{amssymb, array, amscd}

\theoremstyle{plain}
\newtheorem{theorem}{Theorem}[section]

\newtheorem{proposition}[theorem]{Proposition}

\newtheorem{definition}[theorem]{Definition}
\newtheorem{example}[theorem]{Example}

\numberwithin{equation}{section}

\newcommand{\N}{\mathbb{N}}
\newcommand{\Z}{\mathbb{Z}}

\begin{document}
\title[Introducing a probabilistic structure on Sequential Dynamical
Systems ] {Introducing a probabilistic structure on Sequential
Dynamical Systems, Simulation and Reduction of Probabilistic
Sequential Networks}
\author{Maria A. Avino-Diaz}
\address{Gauss Laboratory, Univiversity of Puerto Rico, Rio Piedras,
 Puerto Rico } \email{mavino@uprr.pr}

\thanks{\emph{22000 Mathematics Subject Classification} Primary: 05C20, 37B99, 68Q65, 93A30; Secondary: 18B20, 37B19,
60G99.}
\thanks{$\flat$ We will use
the acronym PBN, PSN, or SDS for plural as well as singular
instances.}
\thanks{$\flat_1$Information giving by Laubenbacher}

\maketitle

\begin{abstract} A probabilistic structure
on sequential dynamical systems is introduced here, the new model
will be called Probabilistic Sequential Network, PSN. The morphisms
of Probabilistic Sequential Networks are defined using two algebraic
conditions. It is proved here that two homomorphic Probabilistic
Sequential Networks have the same equilibrium or steady state
probabilities if the morphism is either an epimorphism or a
monomorphism. Additionally, the proof of the set of PSN with its
morphisms form the category \textbf{ PSN}, having the category of
sequential dynamical systems \textbf{ SDS}, as a full subcategory is
given. Several examples of morphisms, subsystems and simulations are
given.
\end{abstract}

\section{Introduction}\label{intro}
  \par
  Probabilistic Boolean Networks  was introduced by I. Schmulevich,
 E. Dougherty, and W. Zhang in 2000, for studying the dynamic of a network using
  time discrete Markov chains, see \cite{S1, S2,SDZ, S3}. This model had  several applications
  in the study of cancer, see \cite{S4}. It is important for development  an
  algebraic mathematical theory of the model Probabilistic Boolean Network
  PBN, to describe special maps between two PBN,
  called homomorphism and projection, the first papers in this direction are, \cite{DS,ID},
  $\flat$.
  Instead of this model is being used in applications,
  the connection  of the graph of genes and
  the State Space is an interesting problem to study. The introduction of probabilities
  in the definition of  Sequential Dynamical System has this
  objective.

  The theory of sequential dynamical systems (SDS) was born studying
  networks where the entities involved in the problem do not necessarily arrive
  at a place  at the same time, and it is part of the theory of computer
  simulation, \cite{B1, B2}. The mathematical background for SDS was recently development  by
Laubenbacher and Pareigis, and  it solves  aspects of the theory and
applications, see \cite{LP1,LP2, LP3}.

The introduction of a probabilistic structure on Sequential
Dynamical Systems is an  interesting problem that it is introduced
in this paper. A SDS induces a finite dynamical system $(k^n,f)$,
\cite{H}, but the mean difference between a SDS and FDS
  is that a SDS has a graph with information giving by the
  local functions, and an order in the sequential behavior of these
  local functions. It is known  $\flat_1$, that a finite dynamical
  systems can be studied as a SDS, because we can construct a bigger
  system that in this case is sequential.
  Making together the sequential order and the probabilistic structure in the dynamic of the system,
   the possibility to work in applications to genetics increase, because genes act in a
  sequential manner. In particular the notion of morphism in the category of SDS
  establishes connection between the digraph of genes and the State Space, that is the dynamic of the function.
   Working in the applications,  Professor Dougherty's group  wanted to consider two things in the definition of PBN:
  a sequential behavior on genes, and the exact definition of projective maps
  between two PBN that inherits the properties of the first digraph of genes.
   For this reason, a new model that considers both questions and  tries to construct
  projections that  work well is described here. I introduce in this paper
  the sequential behavior and the probability together in PSN and
  my final objective is to construct  projective maps that let us
  reduce the number of functions in the finite dynamical systems
  inside the PBN. One of the mean problem in modeling dynamical systems
  is the computational aspect of the number of functions and the
  computation  of steady states in the State Space. In particular, the reduction of number of functions is one of
the most important problems, because by solving that we can
determine which part of the network \emph{State Space} may be
simplified.  The concept of morphism, simulation, epimorphism, and
equivalent Probabilistic Sequential Networks are developed  in this
paper, with this particular objective.

This paper is organized  as follows. In section \ref{SDS}, a
notation slightly different to the one used in \cite{LP2}  is
introduced for homomorphisms of SDS. This notation is helpful for
giving the concept of morphism of PSN. In section \ref{PSN1}, the
probabilistic structure on SDS is introduced using for each vertex
of the support graph, a set of local functions,  more than one
schedule, and finally having several update functions with
probabilities assigned to them. So, it is obtained a new concept:
probabilistic sequential network (PSN).  In Theorem \ref{MT} is
proved that  monomorphisms, and epimorphisms of  PSN have the same
\emph{equilibrium or steady state probabilities}. These strong
results justify the introduction of the dynamical model PSN as an
application to the study of sequential systems.  In section
\ref{CAT}, we prove that the PSN with its morphisms form the
category \textbf{ PSN}, having the category \textbf{ SDS} as a full
subcategory. Several examples of morphisms, subsystems and
simulations are given in Section \ref{EHOM}.
\section{Preliminaries}\label{SDS}
In this introductory section we give the  definitions and results of
Sequential Dynamical System  introduced  by Laubenbacher and
Pareigis in \cite{LP2}. Let $\Gamma$ be a graph, and let
$V_\Gamma=\{1,\ldots, n\}$ be the set of vertices of $\Gamma$. Let
$(k_i|i \in V_\Gamma)$ be a family of finite sets. The set $k_a$ are
called the set of local states at $a$, for all $a\in V_\Gamma$.
Define $ k^{n} := k_1 \times \cdots \times k_n$ with
  $|k_i|<\infty $, the set of (global) states of $\Gamma$.\\
   A Sequential Dynamical System (SDS)
\[{\mathcal F}=(\Gamma,(k_i)_{i=1}^{n},(f_i)_{i=1}^n, \alpha)
\] consists of
\begin{enumerate}
\item [1.]  A finite graph
$\Gamma=(V_\Gamma,E_\Gamma)$ with the set of vertices
$V_\Gamma=\{1,\ldots ,n\}$ , and the set of edges $E_\Gamma
\subseteq V_\Gamma \times V_\Gamma$.
 \item[2.] A family of finite  sets $(k_i|i\in V_\Gamma)$.
 \item[3.]  A family of  local functions $f_i:
k^n \rightarrow k^n $,  that is  \[f_i(x_1, \ldots , x_n)= (x_1,
\ldots ,x_{i-1}, \overline{f}, x_{i+1}, \ldots , x_n)\] where
$\overline{f}(x_1,\ldots , x_n)$ depends only of those variables
which are connected to $i$ in $\Gamma$.
\item [4.] A
permutation $\alpha =( \begin{array}{lll} \alpha_1& \ldots &
\alpha_n
\end{array}) $ in the set of vertices $V_\Gamma$, called an update
schedule ( i.e. a bijective map $\alpha :V_\Gamma\rightarrow
V_\Gamma)$.
\end{enumerate}
\par The global update function of the SDS is $f=f_{\alpha_1}\circ\ldots\circ{f_{\alpha_n}}$.
 The function $f$ defines the dynamical behavior of the SDS and determines a finite directed graph
  with vertex set $k^n$ and directed edges $(x,f(x))$, for all $x\in {k}^n$, called the State
  Space of $\mathcal{F}$, and denoted by $\mathcal S_f $.\\
The definition of homomorphism between two SDS uses the fact that
the vertices $V_\Gamma=\{1,\ldots , n\}$ of a SDS and the states
$k^n$ together with their evaluation map $k^{n}\times V_{\Gamma} \ni
(x,a) \mapsto <x,a>:=x_a\in k_i$, form a contravariant setup, so
that morphisms between such structures should be defined
contravariantly, i.e. by a pair of certain maps $\phi : \Gamma
\rightarrow \Delta ,$  and the induced function $h_\phi :
k^{m}\rightarrow k^{n}$  with the graph $\Delta$ having $m$
vertices. Here we use  a notation slightly different that the one
using in  \cite{LP2}.

Let $F=(\Gamma,(f_i:k^n\rightarrow {k}^n),\alpha)$ and
$G=(\Delta,(g_i:{k}^m\rightarrow k^m),\beta)$ be two SDS. Let
$\phi:\Delta \rightarrow \Gamma$ be a digraph morphism. Let
$(\widehat{\phi}_b:\overline{}k_{\phi(b)}\rightarrow k_b, \forall
b\in \Delta),$
 be a family of maps in the
category of \textbf{Set}. The map $h_\phi$ is an adjoint map,
because is defined as follows: consider the pairing $k^n\times
V_{\Gamma} \ni (x,a) \mapsto <x,a>:=x_a\in k_a;$ and  similarly $
k^m\times V_\Delta \ni (y,b) \mapsto <y,b>:=y_b\in k_b.$ The induced
adjoint map holds
$<h_\phi(x),b>:=\hat{\phi}_b(<x,\phi(b)>)=\hat{\phi}_b(x_{\phi(b)})$.
Then $\phi,$ and $(\widehat{\phi}_b)$ induce the adjoint map
$h_\phi:k^{n}\rightarrow k^{m}$ defined as follows:
\begin{equation}\label{defh}
           h_\phi(x_1,\ldots,x_n)=(\widehat{\phi}_{1}(x_{\phi(1)}),\ldots,\widehat{\phi}_{m}(x_{\phi(m)})).
\end{equation}

Then $h:F\rightarrow  G$ is a homomorphism of SDS if for all sets of
orders $\tau _{\beta}$ associated to $\beta$ in the connected
components of $\Delta $, the map $h_\phi$ holds the following
conditions:
\begin{equation}\label{eq1}\left( g_{\beta_l}\circ
g_{\beta_{l+1}}\circ \cdots \circ g_{\beta_s} \right)\circ
h_\phi=h_\phi\circ f_{\alpha_i}\hspace{.3in}\end{equation}
\[\begin{CD}
k^n @> f_{\alpha_i} >>k^n \\
@V h_\phi VV            @V h_\phi VV \\
k^m @> g_{\beta_l}\circ \cdots \circ g_{\beta_s}
>> k^m
\end{CD}
\]
where  $\{\beta_l,\beta_{l+1},\ldots,\beta_s\}=\phi
^{-1}(\alpha_i)$. If $\phi^{-1}(\alpha_i)=\emptyset$, then $Id_{k^m}
\circ h_\phi=h_\phi\circ f_{\alpha_i}$, and the commutative diagram
is now the following: \begin{equation}\label{eq2}
\begin{CD}
k^n @ > f_{\alpha_i} >> k^n \\
@V h_\phi VV  @V h_\phi VV \\
k^m @ > Id_{k^m} >>  k^m
\end{CD}
\end{equation}
 For examples and properties see\cite{LP2}. It that paper,
 the authors proved that the above diagrams implies the following one
\begin{equation}\label{eq3}
\begin{CD}
k^n @ > f=f_{\alpha_1}\circ \cdots \circ
f_{\alpha_n} >> k^n \\
@V h_\phi VV  @V h_\phi VV \\
k^m @ > g=g_{\beta_1}\circ \cdots \circ g_{\beta_m} >>  k^m
\end{CD}\end{equation}
\emph{Probabilistic  Boolean Networks} \cite{S1,S2,SDZ,S4}
 The model Probabilistic Boolean  Network
$\mathcal{A}=\mathcal{A}(\Gamma,F,C)$ is defined by the following:
\begin{itemize}
\item [(1)] a finite digraph $\Gamma =(V_\Gamma,
E_\Gamma)$ with  $n$ vertices.
\item [(2)] a family $F=\{F_1, F_2,\ldots,F_n\}$ of  ordered sets
$F_i=\{f_{i1},f_{i2},\ldots,f_{il(i)}\}$ of functions
$f_{ij}:\{0,1\}^n\rightarrow\{0,1\}$, for $i=1,\cdots,n$, and
$j=1,\ldots,l(i)$ called predictors,
\item [(3)]and a family $C=\{c_{ij}\}_{i,j}$, of selection probabilities. The selection probability that the
function $f_{ij}$ is used for the vertex $i$ is $c_{ij}$.
\end{itemize}
\par The dynamic of the model Probabilistic Boolean Network is given
 by the vector  functions $\mathbf
{f}_k=(f_{1k_1}, f_{2k_2},\ldots,f_{nk_n}):\{0,1\}^n\rightarrow
\{0,1\}^n$ for $1\le{k_i}\le{l(i)}$, and $f_{ik_i}\in{F_i}$,  acting
as a transition function.
 Each variable  $x_i\in \{0,1\}$ represents the state of the vertex  $i$.
 All functions are updated synchronously. At every time step, one of the  functions
is selected randomly from the set $F_i$ according to a predefined
probability distribution.  The selection probability that the
predictor $f_{ij}$ is
  used to predict gene $i$ is equal to
\[c_{ij}=P\{f_{ik_i}=f_{ij}\}=\sum_{k_i=j}{p\{{\mathbf
{f}=f}_k\}}.\] There are two digraph structures associated with a
Probabilistic Boolean Network: the low-level graph $\Gamma$, and the
high-level graph which consists of the states of the system and the
transitions between states. The state space $S$ of the network
together with the set of network functions, in conjunction with
transitions between the states and network functions, determine a
Markov chain. The random perturbation makes the Markov chain
ergodic, meaning that it has the possibility of reaching any state
from another state and that it possesses a long-run (steady-state)
distribution. As a Genetic Regulatory Network (GRN), evolves in
time, it will eventually enter a fixed state, or a set of states,
through which it will continue to cycle. In the first case the state
is called a singleton or fixed point attractor, whereas, in the
second case it is called a cyclic attractor. The attractors that the
network may enter depend on the initial state. All initial states
that eventually produce a given attractor constitute the basin of
that attractor. The attractors represent the fixed points of the
dynamical system that capture its long-term behavior. The number of
transitions needed to return to a given state in an attractor is
called the cycle length. Attractors may be used to characterize a
cell's phenotype (Kauffman, 1993) \cite{K}. The attractors of a
Probabilistic Genetic Regulatory Network  (PGRN) are the attractors
of its constituent GRN. However, because a PGRN constitutes an
ergodic Markov chain, its steady-state distribution plays a key
role. Depending on the structure of a PGRN, its attractors may
contain most of the steady-state probability mass \cite{A,MD,Z}.
\section{Probabilistic Sequential Networks}\label{PSN1}
The following definition give us the possibility to have several
update functions acting in a sequential manner with assigned
probabilities. All these, permit us to study the dynamic of these
systems using Markov chains and other probability tools.
\begin{definition}\label{PSN}
 A Probabilistic Sequential Network (PSN) \begin{displaymath}{\mathcal
D}=(\Gamma, \{F_a\}_{a=1}^{|\Gamma|=n}, (k_a)_{a=1}^{n},
(\alpha_j)_{j=1}^m, C=\{c_1,\ldots, c_s\})\end{displaymath} consists
of:
\begin{enumerate}
\item[(1)] a finite graph
  $\Gamma=(V_\Gamma,E_\Gamma)$ with $n$ vertices;

 \item [(2)]  a family of finite  sets $(k_a|a\in V_\Gamma)$.

  \item [(3)] for each vertex $a$ of  $\Gamma$  a  set of
local functions
\[F_a=\{f_{ai}:k^n\rightarrow k^n|1\leq i \leq
\ell(i)\},\] is assigned. (i. e. there exists a bijection map $\sim
: V_\Gamma\rightarrow \{F_a |1\leq a\leq n\}$) (for definition of
local function see (\ref{defSDS}.2)).

\item [(4)] a family of $m$ permutations $\alpha  =(
\begin{array}{lll} \alpha_{1}& \ldots & \alpha_{n}
\end{array}) $ in the set of vertices $V_\Gamma $.

\item [(5)]and  a set $C=\{c_1, \ldots ,c_s\}$, of
assign probabilities to $s$ update functions.
\end{enumerate}
\end{definition}
We select one function  in  each set $F_a$, that is one for each
vertices $a$ of $\Gamma$, and a permutation $\alpha $, with the
order in which the vertex $a$ is selected, so there are
$\underline{n}$ possibly different update functions
$f_i=f_{\alpha_1i_1}\circ \ldots \circ f_{\alpha_ni_n}$, where
$\underline{n}\leq n !\times\ell(1)\times \dots \times \ell (n)$.
The probabilities are assigned to the update functions, so there
exists a set $S=\{f_1,\ldots ,f_s\}$ of selected update functions
such that $c_i=p(f_i)$, $1\leq i \leq s$.
\begin{definition}
 The State Space of ${\mathcal D}$ is a weighted
digraph whose vertices are the elements of $k^{n}$ and there is an
arrow going from the vertex $u$ to the vertex  $v$  if there exists
an update function $f_i\in S$, such that $v=f_i(u)$. The probability
$p(u,v)$ of the arrow going from $u$ to $v$  is the sum of the
probabilities $c_{f_i}$ of all functions $f_i$, such that
$v=f_i(u)$, $u \overset{ p(u,v)}\longrightarrow f_i(u)=v$. We denote
the State Space by $\mathcal {S_D}$.
\end{definition} For each one update function in $S$ we have one  SDS
inside the PSN, so the State Space $\mathcal S_f$ is a subdigraph of
$\mathcal {S_D}$. When we take the whole set of update functions
generated by the data, we will say that we have the \emph{full} PSN.
It is very clear that a SDS is a  particular PSN, where we take one
local function for each vertex, and  one permutation. The dynamic of
a PSN  is described by  Markov Chains of the transition matrix
associated to the State Space.
\begin{example}\label{exam1}  Let ${\mathcal
D}=(\Gamma;F_1,F_2,F_3;\mathbf{Z_2}^{3};\alpha _1, \alpha _2;
(c_{f_i})_{i=1}^8),$ be the following PSN:\\
\emph{(1)} The graph $\Gamma$: $  \ \ \ \begin{array}{c} 1 \bullet  {\overline{\hspace{.2in}}} \bullet \ 3\\
  \diagdown   \mid \\
 2\hspace{.15in}  \bullet  \end{array}.$\\
\emph{(2)}  Let  $\textbf{x}=(x_1,x_2,x_3)\in \{0,1\}^3$. In this
paper, we  always consider the operations over the finite field
$\Z_2=\{0,1\}$, but we use additionally the following notation
$\bar{x_1}=x_1+ 1$. Then the sets of local functions from
${\Z_2}^3\rightarrow {\Z_2}^3$ are the following
\[\begin{array}{l}
F_1 =\{f_{11}(\textbf{x})=(1,x_2,x_3),f_{12}(\textbf{x})=(\bar{x_1},x_2,x_3))\}  \\
F_2 = \{f_{21}(\textbf{x})=(x_1, x_1x_2,x_3)\}   \\
F_3=\{f_{31}(\textbf{x})=(x_1,x_2,x_1x_2),f_{32}(\textbf{x})=(x_1,x_2,x_1x_2+x_3)\}
\end{array}.\]
\emph{(3)} The schedules or permutations are
$\alpha_1=\left(\begin{array}{lll} 3&2& 1
\end{array}\right); \alpha_2=\left(\begin{array}{lll} 1&2& 3
\end{array}\right) .$
  We obtain the following table of
 functions, and we select all of them for ${\mathcal D}$ because the
 probabilities given by $C$.
\[\begin{array}{ll}
f_1=f_{31}\circ f_{21} \circ f_{11} & f_2=f_{11}\circ f_{21} \circ f_{31} \\
f_3=f_{32}\circ f_{21} \circ f_{11} & f_4=f_{11}\circ f_{21} \circ
 f_{32}\\
f_5=f_{31}\circ f_{21} \circ f_{12} & f_6=f_{12}\circ f_{21} \circ
f_{31}\\
f_7=f_{32}\circ f_{21} \circ f_{12}& f_8=f_{12}\circ f_{21}\circ
f_{32} \cr
\end{array}.\]
  The  update functions are the following:
  \[\begin{array}{ll}
  f_1(\textbf{x})=(1,x_2,x_2)& f_2(\textbf{x})=(1,x_1x_2,x_1x_2) \\
 f_3(\textbf{x})=(1,x_2,x_2+x_3) & f_4(\textbf{x})=(1, x_1x_2,x_1x_2+x_3)\\
f_5(\textbf{x})=(\bar{x_1},\bar{x_1}x_2,
 (x_1+1)x_2)  & f_6(\textbf{x})=(\bar{x_1},x_1x_2,x_1x_2)\\
f_7(\textbf{x})=(\bar{x_1},(x_1+1)x_2,(x_1+1)x_2+x_3)&
f_8(\textbf{x})=(\bar{x_1},x_1x_2,x_1x_2+x_3)
 \cr
\end{array}.\]
\emph{(4)} The  probabilities assigned are the following:
$c_{f_1}=.18; c_{f_2}=.12; c_{f_3}= .18; c_{f_4}=.12; c_{f_5}=.12;
c_{f_6}=.08; c_{f_7}=.12; c_{f_8}=.08$.
\end{example}
\begin{example}\label{MA} We notice that there are several PSN that we
can construct with the same initial data of functions and
permutations, but with different set of  probabilities, that is,
subsystems of ${\mathcal D}$. For example if $S'=\{f_1, f_2,
f_3,f_4\}$, $F_1'=\{f_{11}\}$, and $D=\{ d_{f_1}=.355, d_{f_2}=.211,
d_{f_3}=.18,d_{f_4}=.254\}$, then
\[{\mathcal B}=(\Gamma;F_1',F_2,F_3;{\mathbf Z_2}^{3};\alpha _1, \alpha _2; D=\{
.355,.211, .18,.254\} ),\] is a PSN too.
\end{example}
\section{Morphisms of Probabilistic Sequential Networks}\label{secHPSS}
The definition of morphism of PSN is a natural extension of the
concept of homomorphism of SDS. In this section we prove in Theorem
\ref{C3} a strong property, that is the distribution of
probabilities of  two homomorphic PSN are enough close to prove
Theorem \ref{MT}.

Consider the following two PSN ${\mathcal
D}_1=(\Gamma,(F_a)_{a=1}^{|\Gamma|=n},(k_a)_{a=1}^{n},(\alpha
^j)_{j},C)$ and  \\ ${\mathcal D}_2
=(\Delta,(G_b)_{b=1}^{|\Delta|=m},(k_b)_{b=1}^{m},(\beta
^j)_{j},D).$ We denote by $S_i$ the set of update functions of
${\mathcal D}_i$, $i=1,2$; and the following notation for $(u,  v)
\in k^n \times k^n$, and $(w,  z) \in k^m\times k^m$,
\[c_f(u,v)=\left\{\begin{array}{cc}
p(f) & if  \ f(u)=v \\
0 & otherwise  \cr
 \end{array}\right\}, \
 d_g(w,z))=\left\{\begin{array}{cc}
p(g) & if  \ g(w)=z \\
0 & otherwise  \cr
 \end{array}\right\}\]
where $p(h)$ is the probability of the function $h$.
\begin{definition}{\emph{(Homomorphisms of PSN)}}\label{homPSS} A
morphism $h:{\mathcal D}_1\rightarrow {\mathcal D}_2$  consist of:
\begin{enumerate}
\item [(1)]  A graph morphism  $\phi:\Delta \rightarrow \Gamma $, and a family of maps in the category \textbf{Set},
$(\widehat{\phi}_{b}:k_{\phi(b)}\rightarrow k_{b}\forall b\in \Delta
),$ that  induces the adjoint function $h_\phi$, see (\ref{defh}).
 \item [(2)]  The induced adjoint map $h_\phi: k^{n}\rightarrow k^{m}$ holds that
 for all update functions $f$ in $S_1$ there exists an
update function $g\in S_2$ such that $h$ is a SDS-morphism from
$(\Gamma,(f:k^n\rightarrow k^n),\alpha_j)$ to
$(\Delta,(g:k^m\rightarrow k^m),\beta _j)$. That is, the diagrams
\ref{eq1}, \ref{eq2}, and \ref{eq3} commute for all  $f$ and its
selected $g$.
\begin{equation}\label{eq4}
\begin{CD}
k^n @ > f=f_{\alpha_1}\circ \cdots \circ
f_{\alpha_n} >> k^n \\
@V h_\phi VV  @V h_\phi VV \\
k^m @ > g=g_{\beta_1}\circ \cdots \circ g_{\beta_m} >>  k^m
\end{CD}\end{equation}
 \end{enumerate}
\end{definition}
The second condition induces a  map $\mu $ from $S_1$ to $S_2$, that
is $\mu (f)=g$  if the selected function for $f$ is $g$.
 We say that a morphism
$h$ from ${\mathcal D}_1$ to ${\mathcal D}_2$ is a
\textbf{PSN-isomorphism} if $\phi$,  $h_\phi$, and $\mu$ are
bijective functions, and $d(h_\phi(u),h_\phi(g(u))= c(u,f(u))$ for
all $u$, in $k^n$, and all $f\in S_1$, and all $g\in S_2$. We denote
it by ${\mathcal D}_1\cong {\mathcal D}_2$. \vspace{.06in}

\hspace{-.1in}\textsc{Special morphisms.} Let ${\mathcal
D}=(\Gamma,(F_i)_{i=1}^n,(\alpha ^j)_{j\in{J}},C)$ be a PSN.

\hspace{-.1in}\textsc{Identity morphism.}
  The functions $\phi=id_\Gamma,$ $h_\phi=id_{k^n}$, and
$\mu =id_{S}$ define  the \emph{identity morphism} ${\mathcal
I}:\mathcal D\rightarrow \mathcal D$, and it is a trivial example of
a PSN-isomorphism.

\hspace{-.1in}\textsc{Monomorphism} A morphism $h$ of PSN is a
\emph{ monomorphism} if $\phi$ is surjective, $h_{\phi}$ is
injective, and  for all $f$, and its associated $g$ we have that
$d_{g}\leq c_{f}$.

\hspace{-.1in}\textsc{Epimorphism} A morphism is an
\emph{epimorphism} if $\phi$ is injective,  $h_{\phi}$ is
surjective, and for all $f$, and its associated $g$ we have that
$d_{g}\geq c_{f}$.

\hspace{-.1in}\textsc{Remark} If the morphism $h$ is either a
monomorphism or an epimorphism, then the function $\mu$ is not
necessary injective, neither surjective.

\hspace{2in}\textsc{Some theorems}
\begin{theorem}\label{C3} The morphism $h:\mathcal
D_1\rightarrow\mathcal D_2$ induces the following probabilistic
condition:

 For a fixed real number $0\le \epsilon<1$,  the map
$h_\phi$ satisfies the following:
\begin{equation}\label{epsilon}
 max_{u,v}|c_{f}(u,v)-d_{g}(h_\phi(u),h_\phi(v))|\le \epsilon
 \end{equation}
 for all  $f$ in $S_1$, and its selected
$g$ in $S_2$, and all $(u, v)\in k^n \times k^n$.
\end{theorem}
\begin{proof} Suppose $\phi,$ and $h_\phi$ satisfy  the Definition
\ref{homPSS}; and \[|c_{f}(u,v)-d_{g}(h_\phi(u),h_\phi(v))|\ge 1\]
for some $(u,v)\in k^n\times k^n$. Then we have one of the following
cases
\begin{itemize}
\item [1.] $c_{f}(u,v)=1$ and
$d_{g}(h_\phi(u),h_\phi(v))=0$. It is impossible by  condition (2)
in  definition \ref{homPSS}. In fact, if we have an arrow going from
$u$ to $v=f(u)$, then  there exists an arrow going from $h_\phi(u)$
to $h_\phi(v)=g(h_\phi(u))$ by diagram \ref{eq4}, and the
probability $d_{g}(h_\phi(u),h_\phi(v))\ne 0$.
\item[2.] $c_{f}(u,v)=0$, and $d_{g}(h_\phi(u),h_\phi(v))=1$.
 It is impossible because at least there exists  one element
$v_1\in k^n$, such that $f(u)=v_1\in k^{n}$ and $c_{f}(u,v_1)\ne 0$,
then $d_{g}(h_\phi(u),h_\phi(v_1))\ne 0$ too. Since the sum of
probabilities of all arrow going up from $h_\phi(u)$ is equal $1$,
then $d_{g}(h_\phi(u),h_\phi(v))<1$, and our claim holds.
\end{itemize}
Therefore the  condition holds, and always $\epsilon$ exists.
\end{proof}
In the next theorem we will use the following notation:
\begin{itemize}
\item [(1)]  $S_\phi=\mu (S_1)$.
\item [(2)]  $g^{t}=g\circ g \circ \cdots \circ g$, $t$
times.
\item [(3)] $p_t(u,v)=\sum_{f^{t}}c_{f^{t}}(u,v)$, and $p_t(h_\phi(u),h_\phi(v))=\sum_{g^{t}}d_{g^{t}}(h_\phi(u),h_\phi(v))$
\item[(4)] $T_i$ denotes the transition matrix of  ${\mathcal
D}_i$, for $i=1,2$, and  $p_t(u,v)=({T_i}^t)_{(u,v)}$.
\end{itemize}
 \begin{theorem}\label{MT}
If $h:{\mathcal D}_1\longrightarrow{\mathcal D}_2$ is either a
monomorphism or an epimorphism of probabilistic sequential networks,
then:
\[\lim
_{t\rightarrow\infty}|({T_1})^{t}_{u,v}-({T_2})^{t}_{h_\phi(u),h_\phi(v)}|=0,\]
 for all   $(u,v)\in k^n \times k^n$. That is, the equilibrium state
 of both systems are equals.
 \end{theorem}
\begin{proof}
The condition giving by Theorem \ref{C3} asserts that, there exists
a fixed real number $0\le \epsilon<1$,  such that the map $h_\phi$
satisfies the following:
\begin{equation*}
 {\max}_{u,v}|c_{f}(u,v)-d_{g}(\phi(u),\phi(v))|\le \epsilon
 \end{equation*}
 for all  $f$ in $S_1$, and its selected
$g$ in $S_2$, and  all $(u, v)\in k^n \times k^n$.

If there is a function $f$ going from $u$ to $v=f(u)$  in $k ^{n}$,
then there exists a function $g$ going from $h_\phi(u)$ to
$h_\phi(v)$, such that  $g(h_\phi(u))=h_\phi(f(u))$.

Because $c_{f^2}(u,f^2(u))=c_{f}(u,f(u))c_{f}(f(u),f^2(u))={c_f}^2$,
and \[d_{g^2}(h_\phi(u),g^2(h_\phi(u)))=
d_{g}(h_\phi(u),g(h_\phi(u)))d_{g}(g(h_\phi(u)),g^2(h_\phi(u)))={d_g}^2.\]
We have
\[|c_{f^2}(u,f^2(u))-d_{g^2}(h_\phi(u),g^2(h_\phi(u)))|=|{c_f}^2-{d_g}^2|\]

If $h$ is a monomorphism, then $c_f\geq d_g$, for all $f$ and its
associated $g$. Then
\[|c_{f^2}(u,f^2(u))-d_{g^2}(h_\phi(u),g^2(h_\phi(u)))|=|{c_f}^2-{d_g}^2|\leq {c_f}^2.\]
By induction we have that
\[|c_{f^t}(u,f^t(u))-d_{g^t}(h_\phi(u),g^t(h_\phi(u)))|=|{c_f}^t-{d_g}^t|\leq {c_f}^t.\]
 This result implies that
\[|p_t(u,v)-p_t(h_\phi(u),h_\phi(v))|\leq \sum_{f^t}{c_{f}}^t + \delta ^t(u,v)\]
where $\delta ^t(u,v)=\sum_{g\in \overline G(u,v)}{d_{g}}^t,$ and
$\overline G(u,v)=\{g\in G|g(h_\phi(u))=h_\phi(v),\ and \ g\ne
\overline \mu(S_1)\}$.

Then, when $t$ goes to infinity the sum $\sum_{f^t}{c_{f}}^t +
\delta ^t(u,v)$ goes to $0$, and the theorem holds. If $h$ is an
epimorphism we obtain the same results, so the theorem holds again.

 \end{proof}
\section{The category  \textbf{PSN}}\label{CAT}
In this section, we  prove that the PSN with the \emph{morphisms}
form a category, that we denote by \textbf{PSN}. For definitions,
and results in Categories see  \cite{ML}.
\begin{theorem}\label{COM} Let
$ h_1=(\phi _1, h_{\phi _1}):{\mathcal D}_1\rightarrow {\mathcal
D}_2$ and   $h_2=(\phi _2, h_{\phi _2}):{\mathcal D}_2\rightarrow
{\mathcal D}_3$ be two morphisms of PSN. Then the composition
$h=(\phi , h_{\phi })=(\phi _2, h_{\phi _2}) \circ (\phi _1, h_{\phi
_1})=h_2\circ h_1: {\mathcal D}_1\rightarrow {\mathcal D}_3$  is
defined as follows: $h=(\phi , h_{\phi })=(\phi_1 \circ \phi _2,
h_{\phi _2} \circ h_{\phi _1})$ is  a morphism of PSN. The function
$\mu_h=\mu_{h_2}\circ \mu_{h_1}.$
\end{theorem}
\begin{proof}
The composite function $\phi=\phi_1 \circ \phi_2 $ of two graph
morphisms is again a graph morphism. The composite function
$h_\phi=h_{\phi _2} \circ h_{\phi _1}$ is again a digraph morphism
which satisfies the conditions in Definition \ref{homPSS}, by
Proposition and Definition 2.7 in \cite{LP2}.  So, $h=(\phi,
h_\phi)$ is again a morphism. of PSN.
\end{proof}
\begin{theorem}
The Probability Sequential Networks together with the homomorphisms
of PSN form the category \textbf{PSN}.
\end{theorem}
\begin{proof} Easily  follows from Theorem \ref{COM}.
\end{proof}
\begin{theorem}\label{FSDS} The SDS together with the morphisms defined in \cite{LP2} form a full subcategory of the category
\textbf{PSN}.
\end{theorem}
\begin{proof}
It is trivial.
\end{proof}
 \section{Simulation and  examples}\label{EHOM}
In this section we give several examples of morphisms, and
simulations. In the second example we show how the Definition
\ref{homPSS} is verified under the supposition that a function
$\phi$ is defined. So, we have two examples in (\ref{EHOM}.2), one
with $\phi$ the natural inclusion, and the second  with $\phi$  a
surjective map. The third, and the fourth examples are morphisms
that represent simulation of $\mathcal G$ by $\mathcal F$. We begin
this section with the definitions of Simulation and sub-PSN.

\hspace{-.15in}\textsc{Definition of Simulation in the category
\textbf{ PSN}}.
 The probabilistic sequential network $\mathcal G$ is simulated by $\mathcal F$ if
there exists a  monomorphism $h: {\mathcal F}\rightarrow {\mathcal
G}$ or an epimorphism $h':{\mathcal G}\rightarrow {\mathcal F}$.

\hspace{-.15in}\textsc{Sub Probabilistic Sequential
Network}\label{subPSS}
 We say that a PSN ${\mathcal G}$ is a
sub Probabilistic Sequential Network of ${\mathcal F}$ if there
exists a monomorphism from ${\mathcal G}$ to ${\mathcal F}$. If the
map $\mu$ is  not a bijection, then we say that it is a proper
sub-PSN.

\subsection{Examples}\label{examples}\hspace{2.2in}

 \textbf{(\ref{examples}.1)} In the examples \ref{exam1}, and \ref{MA} we
define  two PSN $\mathcal D$ and $\mathcal B$. The functions $\phi
=Id_\Gamma$, $h_\phi=Id_{{k}^{n}}$, and $\mu$ the natural inclusion
from $S_1$ to $S_2$ define the inclusion $\iota_\mu: {\mathcal
B}\rightarrow {\mathcal D}$. It is clear that the inclusion is a
monomorphism, so  $\mathcal D$ is simulated by ${{\mathcal B}}$.

 \textbf{(\ref{examples}.2)} Consider the two
graphs below
\[\hspace{.3in}\Gamma \ \
\begin{array}{ccc} 2 \bullet&{\overline{\hspace{.2in}}} &\bullet 3\\
\mid &   &\mid  \\
1 \bullet& {\overline{\hspace{.2in}}} &\bullet 4
\end{array}\hspace{.3in} \hbox{   and       }\hspace{.3in} \Delta \  \  \begin{array}{ccc} 2 \bullet & {\overline{\hspace{.2in}}} &\bullet 3   \\
\mid &   & \\
1\bullet& &
\end{array}\]
 Suppose that the functions
associated to the vertices are the  families $\{f_1,f_2,f_3,f_4\}$,
 for $\Gamma $ and $\{g_1,g_2,g_3\}$ for $\Delta $. The
permutations are $\alpha_1=(4\ 3\ 2\ 1)$, $\alpha_2=(4\ 1\ 3\ 2)$
and $\beta_1=(3\ 2\ 1)$, $ \beta_2=(1\ 3\ 2)$, so, $S=\{f=f_4\circ
f_3 \circ f_2 \circ f_1; \underline{f}=f_1\circ f_4 \circ f_3 \circ
f_2\}$, and $X=\{g= g_3 \circ g_2 \circ g_1; \underline{g}= g_1
\circ g_3 \circ g_2\}$. Then, we  have constructed two PSN, each one
with two permutations and only one function associated to each
vertex in the graph; denoted by:
\[{\mathcal D}=(\Gamma;f_1,f_2,f_3,f_4; \alpha_1, \alpha_2;C ) \ and \
{\mathcal B}=(\Delta;g_1,g_2,g_3;\beta_1, \beta_2;D ).\]
\textbf{Case (a)} We assume that there exists a morphism
$h:{\mathcal D}\rightarrow{\mathcal B}$, with the graph morphism
$\phi: \Delta \rightarrow \Gamma $  giving by $\phi(1)=1,\
\phi(2)=2, \ \phi(3)=3$. Suppose the functions \[(\widehat{\phi}_b
:{k}_{\phi(b)}\rightarrow k_b, \forall b\in \Delta),\] are giving,
and the adjoint function
\[h_{\phi}:k^4\rightarrow k^3,
 \  h_\phi (x_1,x_2,x_3,x_4)=(\hat{\phi}_1(x_1),\hat{\phi}_2(x_2),\hat{\phi}_3(x_3))\] is defined too.
If $h$ is a morphism, which satisfies the definition (\ref{homPSS}),
then the following diagrams commute:
\[\begin{CD}
k^4 @ >f_4 >> k^4 @ > f_3 >> k^4 @ >f_2>> k^4
@ >f_1>> k^4  \\
@V h_{\phi} VV  @V h_{\phi} VV @V h_{\phi} VV @ V h_{\phi} VV @ V h_{\phi} VV \\
k^3 @ > Id >> k^3 @ > g_3 >> k^3 @ > g_2 >> k^3 @
> g_1 >> k^3
\end{CD}, \]
\[\begin{CD}
k^4 @ >f_1 >> k^4 @ > f_4 >> k^4 @ >f_3>> k^4
@ >f_2>> k^4  \\
@V h_{\phi} VV  @V h_{\phi} VV @V h_{\phi} VV @ V h_{\phi} VV @ V h_{\phi} VV \\
k^3 @ > g_1 >> k^3 @ > Id >> k^3 @ > g_3 >> k^3 @
> g_2 >> k^3
\end{CD}, \]

 \[\begin{CD}
 k^4 @ >f >> k^4\\
 @V h_{\phi} VV  @V h_{\phi} VV \\
k^3 @ > g >> k^3
  \end{CD}\hspace{.2in} \begin{CD} k^4 @ > \underline f >> k^4\\
 @V h_{\phi} VV  @V h_{\phi} VV \\
 k^3 @ > \underline g >> k^3
  \end{CD}.\]
\textbf{Case (b)} Consider now the map $\phi: \Gamma \rightarrow
\Delta$, defined by $\phi(1)=1$, $\phi(2)=2$, $\phi(3)=3$, and
$\phi(4)=3$. If there exists a morphism  $h : {\mathcal
B}\rightarrow {\mathcal D}$ that satisfies Definition \ref{homPSS},
then
\[\begin{CD}
 k^3 @ > g_3 >> k^3 @> g_2 >> k^3 @ > g_1 >> k^3  \\
 @V h_{\phi} VV  @V h_{\phi} VV @V h_{\phi} VV @ V h_{\phi} VV  \\
 k^4 @> f_4\circ f_3 >>  k^4 @ > f_2 >> k^4 @ > f_1 >> k^3
\end{CD},  \]
\[\begin{CD}
 k^3 @ > g_1 >> k^3 @> g_3 >> k^3 @ > g_2 >> k^3  \\
 @V h_{\phi} VV  @V h_{\phi} VV @V h_{\phi} VV @ V h_{\phi} VV  \\
 k^4 @>  f_1 >>  k^4 @ > f_4\circ f_3 >> k^4 @ > f_2 >> k^3
\end{CD},  \]
\textbf{(\ref{examples}.3)} We now construct a monomorphism $h:
\mathcal F\rightarrow \mathcal G$, with the properties that $\phi$
 is surjective and the function $h_\phi$ is  injective.
   The PSN ${\mathcal F}=(\Gamma ,(F_i
)_3,\beta ,C )$ has the support graph $\Gamma$ with three vertices,
and the PSN ${\mathcal
 G}=(\Delta ,(G_i )_4,\alpha,D)$ has the support graph $\Delta $ with four
vertices
\[\hspace{.3in} \Gamma \ \hspace{.3in} \ \begin{array}{ccc}
&  & \bullet \ 3\\
&  & \mid \\
 1\ \bullet & {\overline{\hspace{.2in}}} &\bullet \ 2
\end{array} \hspace{.3in}\ \  \hspace{.4in} \Delta \ \ \hspace{.3in} \ \ \
\begin{array}{ccc}
 2\ \bullet &{\overline{\hspace{.2in}}} & \bullet \ 4\\
& \diagdown  & \mid \\
 1\ \bullet & {\overline{\hspace{.2in}}} & \bullet \ 3
\end{array}  \hspace{.2in}   \    \]
The morphism $h:{\mathcal F} \rightarrow {\mathcal G}$, has the
contravariant  graph morphism $\phi : \Delta\rightarrow \Gamma$,
defined by the arrows of graphs, as follows $\phi(1)=1$,
$\phi(2)=\phi(3)=2,$ and $\phi(4)=3 $, so it is a surjective map.
The family of functions $\hat{\phi}_i:{k}_{\phi (i)}\rightarrow
 k_{(i)}$,  $\hat{\phi}_1(x_1)=x_1$; $\hat{\phi}_2(x_2)=x_2;\
 \hat{\phi}_3(x_2)=x_2;\  \hat{\phi}_4(x_4)=x_4$, are injective functions.
 The sets $k_a=\Z_2$, for all vertices $a$ in $\Delta$, and $\Gamma$. The adjoint function is
 $h_\phi:{\mathbf  Z_2}^3\rightarrow {\mathbf  Z_2}^4,$  defined by
  \[  h_\phi(x_1,x_2,x_3)=(\hat{\phi}_1(x_1),\hat{\phi}_2(x_2),\hat{\phi}_3(x_2),\hat{\phi}_4(x_4))=(x_1,x_2,x_2,x_3).\]
 Then, the first condition in the definition \ref{homPSS} holds.

The  PSN  ${\mathcal F}=(\ \Gamma ;\ (F_i)_3 ;\ \beta  ;\ C
),$ is defined with the following data.\\
The set of  functions $ F_1=\{ \ f_{11}, \ f_{12})\},$
 $F_2=\{f_{21}\},$ and $
F_3=\{f_{31}\},$ where
\[f_{11}=Id,\ f_{12}(x_1,x_2,x_3)=(1,x_2,x_3), \
f_{21}=Id, \]
\[f_{31}(x_1,x_2,x_3)=(x_1,x_2,x_2\overline{x_3}).\]
A permutation  $\beta =(\ 1 \ 2 \ 3\ );$ and the probabilities
$C=\{c_{f_1}=.5168, c_{f_2}=.4832\} $. So, we are taking all the
update functions   $S =\{f_1,f_2\}$;
 \[f_1= f_{11}\circ f_{21}\circ f_{31},\
f_1=(x_1,x_2,x_3)= (x_1,x_2,x_2\overline{x_3});\]
\[\textrm{and} \ \ \ f_2= f_{12}\circ f_{21}\circ f_{31}, \
f_2(x_1,x_2,x_3)= (1,x_2,x_2\overline{x_3}).\]

On the other hand, the PSN ${\mathcal G}=(\Delta ;(G_i)_4;\alpha; D )$ has the following data.\\
The families of functions: $G_1=\{g_{11}, g_{12}\}$; $G_2=\{g_{21},
g_{22}\}$, $G_3=\{g_{31},g_{32}\}$; and $G_4=\{g_4\}$, where
\[\begin{array}{ll}
g_{11}(x_1,x_2,x_3,x_4)&= (1,x_2,x_3,x_4) \\
g_{21}(x_1,x_2,x_3,x_4)&= (x_1,1,x_3,x_4)  \\
g_{31}(x_1,x_2,x_3,x_4)&= (x_1,x_2,x_1x_2,x_4) \\
g_{41}(x_1,x_2,x_3,x_4)&=(x_1,x_2,x_3,x_2\overline{x_4})\\
\end{array}\ \begin{array} {ll}
g_{12}= Id=g_{22} &  \\
g_{32}(x_1,x_2,x_3,x_4)&= (x_1,x_2,x_2,x_4)\\
\end{array}.\]
One permutation or schedule  $\alpha=\left( \begin{array}{llll} 1&2&
3& 4\end{array}\right)$. The assigned probabilities $d_{g_5}=
.00252,\ d_{g_6}=.08321,\  d_{g_7}=.51428,\  d_{g_8}=.39999$ whose
determine the set of update functions $X=\{g_5,g_6,g_7,g_8\}$:
therefore the all update functions are  the following
\[\begin{array}{lll}
g_1=g_{11} \circ g_{21} \circ g_{31 }\circ g_{41} ,& g_2= g_{12}
\circ g_{21} \circ g_{32}\circ g_{41} &
g_3=g_{12} \circ g_{21} \circ g_{31}\circ g_{41},\\
 g_4=g_{11}\circ g_{21}\circ g_{32}\circ g_{41}& g_5=g_{12} \circ g_{22} \circ g_{31}\circ g_{41},& g_6=g_{11} \circ g_{22} \circ g_{31} \circ g_{41}\\
g_7=g_{12} \circ g_{22} \circ g_{32}\circ g_{41},& g_8=g_{11}\circ
g_{22}\circ g_{32}\circ g_{41} &\cr
\end{array}.\]
The selected functions are
\[\begin{array}{ll}
g_5(x_1,x_2,x_3,x_4)=(x_1,x_2,x_1x_2,x_2\overline{x_4}),&
g_6(x_1,x_2,x_3,x_4)=(1,x_2,x_1x_2,x_2\overline{x_4})\\
g_7(x_1,x_2,x_3,x_4)=(x_1,x_2,x_2,x_2\overline{x_4}),&
g_8(x_1,x_2,x_3,x_4)=( 1,x_2,x_2,x_2\overline{x_4})\\
\end{array}.\]
  We claim that $h: \mathcal F \rightarrow \mathcal G$ is a
  morphism, in fact the following diagrams commute.
\[\begin{CD}
{\Z_2}^3 @ > f_1 >>{\Z_2}^3 \\
 @V h_{\phi} VV  @V h_{\phi} VV  \\
{\Z_2}^4@ > g_7 >> {\Z_2}^4
\end{CD}, \ \textrm{ and  } \  \ \begin{CD}
{\Z_2}^3@> f_2 >> {\Z_2}^3 \\
 @V h_{\phi} VV  @V h_{\phi} VV \\
 {\Z_2}^4 @> g_8 >> {\Z_2}^4
\end{CD}.\]
In fact, $(h_\phi \circ f_1)(x_1, x_2,x_3)
=(x_1,x_2,x_2,x_2\overline{x_3})=(g_7\circ h_\phi)(x_1, x_2,x_3),$
 on the other hand $(h_\phi \circ
f_2)(x_1,x_2,x_3)=(1,x_2,x_2,x_2\overline{x_3})=(g_8\circ
h_\phi)(x_1,x_2,x_3)$ so, the property holds.  We verify the second
property in the definition of morphism  for the compositions $f_1$
and $g_7$, and also with the compositions $f_2$ and $g_8$. That is,
we check the sequence of local functions too.
\[\begin{CD}
  {\Z_2}^3@ > f_{31}>> {\Z_2}^3@ > f_{21}>>{\Z_2}^3 @> f_{11}>>{\Z_2}^3  \\
 @V h_{\phi}VV  @ V h_{\phi} VV @V h_{\phi} VV @V h_{\phi} VV  \\
 {\Z_2}^4 @ > g_{41} >> {\Z_2}^4 @> g_{22}\circ g_{32}>> {\Z_2}^4  @ > g_{12} >> {\Z_2}^3
\end{CD}\]
\[\begin{CD}
  {\Z_2}^3@ > f_{31}>> {\Z_2}^3@ > f_{21}>>{\Z_2}^3 @> f_{12}>>{\Z_2}^3  \\
 @V h_{\phi}VV  @ V h_{\phi} VV @V h_{\phi} VV @V h_{\phi} VV  \\
 {\Z_2}^4 @ > g_{41} >> {\Z_2}^4 @> g_{22}\circ g_{32}>> {\Z_2}^4  @ > g_{11} >> {\Z_2}^3
\end{CD}\]
$\begin{array}{cl}
 &(h_\phi \circ
f_{31})(x_1,x_2,x_3)=\ (x_1,x_2,x_2,x_2\overline{x_3}) \ =
(g_{41}\circ h_\phi)(x_1,x_2,x_3), \\
&(h_\phi \circ f_{21})(x_1,x_2,x_3)=\ (x_1,x_2,x_2,x_3) \ = ((g_{22}
\circ g_{32}) \circ
h_\phi)(x_1,x_2,x_3),\\
&(h_\phi \circ f_{11})(x_1,x_2,x_3)=\ (x_1,x_2,x_2,x_3)\ =
(g_{12}\circ h_\phi)(x_1,x_2,x_3),\\
&(h_\phi \circ f_{12})(x_1,x_2,x_3)= \ (1,x_2,x_2,x_3)= \
(g_{11}\circ
h_\phi)(x_1,x_2,x_3)\\
\end{array}$.\\
The probabilities satisfy the following conditions: $p(f_1)\geq
p(g_7)$, and $p(f_2)\geq p(g_8)$. Then our claim holds, and $h_\phi$
is a monomorphism.

 \textbf{(\ref{examples}.4)} We can construct
an epimorphism $h':\mathcal G\rightarrow \mathcal F$,  that is, the
function $\phi$ is injective and the function $h'_{\phi}$ is
surjective. We use $\phi':\Gamma\rightarrow\Delta$, defined as
follow $\phi'(i)=i+1$, for all $i\in V_\Gamma$. Therefore
$\hat{\phi'}_i:{k}_{\phi'(i)}\rightarrow k_{(i)}$,
$\hat{\phi'}_i:\Z_2\rightarrow  \Z _2$, for all $i\in V_\Gamma $,
and should be satisfies
$<h_\phi(x),i>:=\hat{\phi}_b(<x,\phi(i)>)=\hat{\phi}_b(x_{\phi(i)})$.
So, the adjoint function is
$h'_\phi(x_1,x_2,x_3,x_4)=(\hat{\phi'}_1(x_1),\hat{\phi'}_2(x_3),\hat{\phi'}_3(x_3))=(x_1,x_2,
x_4)$and  satisfies the following commutative diagrams
\[\begin{CD}
{\Z_2}^4 @ > g_5 >>{\Z_2}^4 \\
 @V h'_{\phi} VV  @V h'_{\phi} VV  \\
{\Z_2}^3@ > f_1 >> {\Z_2}^3
\end{CD}, \ \  \begin{CD}
{\Z_2}^4 @ > g_7 >>{\Z_2}^4 \\
 @V h'_{\phi} VV  @V h'_{\phi} VV  \\
{\Z_2}^3@ > f_1 >> {\Z_2}^3
\end{CD}, \ \  \begin{CD}
{\Z_2}^4@> g_6 >> {\Z_2}^4 \\
 @V h'_{\phi} VV  @V h'_{\phi} VV \\
 {\Z_2}^3 @> f_2 >> {\Z_2}^3
\end{CD} \textrm{ and  } \ \  \ \begin{CD}
{\Z_2}^4@> g_8 >> {\Z_2}^4 \\
 @V h'_{\phi} VV  @V h'_{\phi} VV \\
 {\Z_2}^3 @> f_2 >> {\Z_2}^3
\end{CD}.\]
These implies that $\mu(g_5)=\mu(g_7)=f_1$, and
$\mu(g_6)=\mu(g_8)=f_2$. In fact,
\[(h'_\phi \circ g_5)(x_1,
x_2,x_3,x_4)=(x_1,x_2,x_2\overline{x_4})=(f_1\circ
h'_\phi)(x_1,x_2,x_3,x_4),\]
\[(h'_\phi \circ g_6)(x_1,x_2,x_3,x_4)=(1,x_2,x_2\bar{x_4})=(f_2\circ
h'_\phi)(x_1,x_2,x_3,x_4), \]
\[(h'_\phi \circ g_7)(x_1,
x_2,x_3,x_4)=(x_1,x_2,x_2\overline{x_4})=(f_1\circ
h'_\phi)(x_1,x_2,x_3,x_4),\]
\[(h'_\phi \circ g_8)(x_1,
x_2,x_3,x_4)=(1,x_2,x_2\bar{x_4})=(f_2\circ
h'_\phi)(x_1,x_2,x_3,x_4).\]

Checking the compositions of local functions $g_5=g_{12} \circ
g_{22} \circ g_{31}\circ g_{41}$, and \\$f_1= f_{11}\circ
f_{21}\circ f_{31}$, we have that the following diagrams commute
\[\begin{CD}
  {\Z_2}^4@ > g_{12}>> {\Z_2}^4@ > g_{22}>>{\Z_2}^4 @> g_{31}>>{\Z_2}^4@> g_{41}>> {\Z_2}^4  \\
 @V h'_{\phi}VV  @ V h'_{\phi} VV @V h'_{\phi} VV @V h'_{\phi} VV @V h'_{\phi} VV \\
 {\Z_2}^3 @ > f_{11} >> {\Z_2}^3 @> Id>> {\Z_2}^3  @ > f_{21} >>
 {\Z_2}^3 @> f_{31} >>{\Z}^2\cr
\end{CD}.\]

By the data we only need to check the following compositions \\
$h'_\phi(g_{31}(x_1,x_2,x_3,x_4))=(x_1,
x_2,x_4) =f_{21}(h'_\phi(x_1,x_2,x_3,x_4)),$\\
$h'_\phi(g_{41}(x_1,x_2,x_3,x_4))=(x_1, x_2,x_2\bar{x_4})
=f_{31}(h'_\phi(x_1,x_2,x_3,x_4)).$ Similarly, we can prove that the
other functions hold the property. The probabilities satisfy the
following conditions: $p(g_5)\leq p(f_1)$,$p(g_7)\leq
p(f_1)$,$p(g_6)\leq p(f_2)$, and $p(g_8)\leq p(f_2)$. Therefore
$h'_\phi$ is an epimorphism.
\section{ Equivalent  Probabilistic Sequential Networks}\label{epsn}
\begin{definition}{\emph{(Equivalent PSN)}}\label{equiv} If the  morphism  $h:\mathcal D_1 \rightarrow \mathcal D_2$  satisfies
that $\phi$,  $h_\phi$  and $\mu$ are bijective functions, but the
probabilities are not necessary equals, we say that ${\mathcal
D}_1$, and ${\mathcal D}_2$ are equivalent PSN. We write ${\mathcal
D}_1 \simeq  {\mathcal D}_2$.
\end{definition}
So, ${\mathcal D}_1$, and ${\mathcal D}_2$ are equivalents if there
exist $(\phi,h_\phi,\mu)$, and $(\phi^{-1}, h_\phi^{-1},\mu ^{-1})$,
such that for all update functions $f\in {\mathcal D}_1$ and its
selected function $g\in {\mathcal D}_2$, the condition
$f=h_\phi^{-1}\circ g \circ h_\phi$ holds . It is clear that this
relation is an equivalence relation in the set of PSN.
\begin{proposition} If ${\mathcal D}_1 \simeq  {\mathcal D}_2$, then the
transition matrices $T_1$ and $T_2$ satisfy: $(T_1^m)_{(u,v)}\ne 0$,
if and only if $(T_2^m)_{(h_\phi{(u)},h_\phi(v))}\ne 0$, for all
$m\in \N$, $(u,v)\in k^n\times k^n$.
\end{proposition}
\begin{proof} It is obvious.
\end{proof}
\section*{Acknowledgment} This work  was supported by the Partnership M. D.
Anderson Cancer Center, University of Texas, and the Medical
Sciences Cancer Center, University of Puerto Rico, by  the Program
AABRE of Rio Piedras Campus,  and by the SCORE Program of NIH, Rio
Piedras Campus.

\end{document}